\documentclass[11pt, a4paper]{article}

\usepackage{a4wide}
\usepackage{amsmath, amsthm, amssymb, hyperref}
\usepackage{filecontents}
\usepackage{mathtools}
\usepackage{thm-restate}
\usepackage{enumerate}
\usepackage{color, soul}

\hypersetup{colorlinks=true, citecolor=blue, linkcolor=blue, urlcolor=blue}

\newtheorem{theorem}{Theorem}
\newtheorem{lemma}[theorem]{Lemma}
\newtheorem{corollary}[theorem]{Corollary}
\theoremstyle{definition}
\newtheorem{definition}[theorem]{Definition}

\title{Fast algorithms for general spin systems on bipartite expanders
\footnote{A preliminary short version of the manuscript (without proofs) will appear in the proceedings of MFCS 2020.}}

\author{Andreas Galanis
\thanks{Authors' address: Department of Computer Science, University of Oxford, Wolfson Building, Parks Road, Oxford, OX1~3QD, UK.}
\and
Leslie Ann Goldberg$^*$
\and
James Stewart$^*$}

\date{14 April 2020}

\def\calG{\mathcal{G}}
\def\calC{\mathcal{C}}
\def\calP{\mathcal{P}}
\def\calK{\mathcal{K}}

\def\Gammab{\ensuremath{\boldsymbol{\Gamma}}}
\def\sigmab{\ensuremath{\boldsymbol{\sigma}}}

\def\overlap{\ensuremath{\mathsf{overlap}}}

\def\Bb{\ensuremath{\mathbf{B}}}

\def\wBb{\ensuremath{\widehat{\Bb}}}
\def\epsilon{\varepsilon}
\def\spin{\text{\normalsize{S}\scriptsize{PIN}}}
\def\hmax{h_{\text{max}}}

\newcommand{\floor}[1]{\left\lfloor #1 \right\rfloor}

\newcommand{\gbip}{\calG^\text{bip}_{ \Delta, \lambda}}
\newcommand{\kmax}{\calK_H^\text{max}}

\newcommand{\prob}[4]{
\hspace*{0.5ex}
\vbox{
\begin{description}
  \item[\bf Parameters:] #1
  \vspace{-1.75ex}
  \item[\bf Name:] #2
  \vspace{-1.75ex}
  \item[\bf Input:] #3
  \vspace{-1.75ex}
  \item[\bf Output:] #4
\end{description}
}
}

\begin{document}
\maketitle{}

\begin{abstract}
A spin system is a framework in which the vertices of a graph are assigned spins from a finite set. 
The interactions between neighbouring spins give rise to weights, so a spin assignment can also be viewed as a weighted graph homomorphism.
The problem of approximating the partition function (the aggregate weight of spin assignments)
or of sampling from the resulting probability distribution is typically intractable for general graphs.

In this work, we consider arbitrary spin systems on  bipartite expander $\Delta$-regular graphs, including the canonical class of bipartite random $\Delta$-regular graphs. We develop fast approximate sampling and counting algorithms for general spin systems whenever the degree and the spectral gap of the graph are sufficiently large. Roughly, this guarantees that the spin system is in the so-called low-temperature regime. Our approach generalises the techniques of Jenssen et al. and Chen et al.  by showing that typical configurations on bipartite expanders correspond to ``bicliques'' of the spin system; then, using suitable polymer models, we show how to sample such configurations and approximate the partition function in $\tilde{O}(n^2)$ time, where $n$ is the size of the graph. 
\end{abstract}
 
\section{Introduction}

Spin systems are general frameworks that encompass  sampling and counting problems in computer science, graph homomorphism problems in combinatorics, and phase transition phenomena in statistical physics.  In this paper, we provide algorithms for general spin systems on bounded-degree bipartite expander graphs. 

A $q$-spin system is specified by a set of spins $[q]=\{1,2,\hdots,q\}$ and a symmetric interaction matrix $H \in \mathbb{R}_{\geq 0}^{q \times q}$. Given a graph $G = (V_G, E_G)$, a spin configuration is an assignment $\sigma : V_G \rightarrow [q]$ and  the weight of $\sigma$ is given by $w_{G, H}(\sigma) = \prod_{\{u, v\} \in E_G} H_{\sigma(u), \sigma(v)}$. The Gibbs distribution of the system, denoted by $\mu_{G,H}$, is a probability distribution on the set $\Sigma_{G, H}$ which denotes the set of all possible spin configurations, given by 
\[\mu_{G,H}(\sigma)=\frac{ w_{G, H}(\sigma)}{Z_{G,H}},\]
where $Z_{G, H} := \sum_{\sigma \in \Sigma_{G, H}} w_{G, H}(\sigma)$ is the so-called partition function.  Well-known examples of spin systems are the Ising/Potts models, where the matrix $H$ has all diagonal entries equal to some parameter $\beta>0$ and off-diagonal entries equal to 1; the case $q=2$ is the Ising model, and $q>2$ is the Potts model. Apart from statistical physics systems, graph homomorphisms also fit naturally into this framework, whenever $H$ has 0-1 entries.

Henceforth, we will normalise $H$ so that its largest entry is equal to $1$. More formally, we will consider $\delta$-matrices, defined as follows.
\begin{definition}
\label{def:deltamatrix}
Let $q \geq 2$ and let $\delta \in (0, 1)$. A symmetric matrix $H \in \mathbb{R}_{\geq 0}^{q \times q}$ is called a $\delta$\textit{-matrix} if $\max_{i, j \in [q]} H_{i, j} = 1$ and $\max_{i, j \in [q] : H_{i, j} \neq 1} H_{i, j} \leq \delta$.\qed
\end{definition}

Note that a every symmetric $0$-$1$ matrix is a $\delta$-matrix
for every $\delta\in(0,1)$. Also, apart from trivial cases\footnote{If $H$ has all of its entries equal to some $c>0$, then $Z_{G, H}=q^{|V_G|}c^{|E_G|}$ for any graph $G=(V_G,E_G)$. Similarly, if $H$ is the all zeros matrix, then $Z_{G, H}=0$ for any non-empty graph $G$.  So suppose that $H$ has at least two entries with distinct values and let $\hmax = \max_{i, j \in [q]} H_{i, j}$. Then, the matrix $H'=\frac{1}{\hmax}H$ is a $\delta$-matrix, for any $\delta\in (0,1)$ which is bigger than the second largest entry in $H'$. Moreover, for any graph $G=(V_G,E_G)$ we have that $Z_{G, H} = \hmax^{|E_G|} \cdot Z_{G, H'}$.}, we can always normalise the interaction matrix of any spin system to satisfy Definition~\ref{def:deltamatrix} for some $\delta\in (0,1)$.

Approximately sampling from the Gibbs distribution of a spin system and approximating its partition function are well-studied problems in computer science, since they appear in various applications. However, even for the most canonical models, such as the Potts model or graph homomorphisms, these computational problems are hard in general, even on bounded-degree graphs \cite{Greenhill, Bulatov, govorov2020dichotomy, sampling, galanis2016approximately, Potts, ferro}. 

 
Despite these hardness results for specific models, there is currently no known characterisation classifying the complexity of approximating $Z_{G,H}$
(determining for which~$H$ approximation is tractable).
However, there are some hardness results which apply to broad classes of~$H$. 
For example, 
it has been shown    \cite{galanis2016approximately}  
that if $H$ is the adjacency matrix of any (non-trivial) bipartite graph
(i.e., a bipartite graph whose connected components are not all complete)
then approximating  $Z_{G,H}$ is \#BIS-hard\footnote{\#BIS is the complexity
class containing all approximate counting problems
that are equivalent to approximately counting the independent sets of a bipartite graph. It is an important class in the field of approximate counting, but
we will not need more details in this paper.
This hardness result does not represent a complexity classification,
even in the $0$-$1$ case. In fact, 
there are known examples \cite[Theorem 5.1]{trees} 
of $0$-$1$ matrices~$H$ such that $Z_{G,H}$ is  NP-hard, even when $G$ is restricted to be bipartite. } , even when the input~$G$
is restricted to be bipartite.


In light of these hardness results, it is natural to consider whether efficient algorithms can be developed for more restricted classes of graphs.  Recently, Jenssen, Keevash, and Perkins \cite{JKP} (see also \cite{helmuth2019algorithmic}) showed a new framework that such algorithms are possible for three canonical models (ferromagnetic Potts model, colourings, and independent sets) on bounded-degree expander graphs, in a range of parameters where the problems are otherwise hard for  general bounded-degree graphs.  
See also independent applications to colourings and independent sets in~\cite{liao2019counting}.

In this paper, we show that this framework can further be used to obtain approximation algorithms of \emph{any} spin system on \emph{bipartite} expander graphs whenever the degree is sufficiently large.

More precisely, we will consider regular bipartite graphs whose second eigenvalue is bounded by a small constant \cite{hoory2006expander}. Let $G $ be an $n$-vertex $\Delta$-regular bipartite graph. Let $\lambda_1(G) \geq \lambda_2(G) \geq \dots \geq \lambda_n(G)$ denote the eigenvalues of the adjacency matrix of $G$. It is well-known that $\lambda_1(G) = \Delta$ and $\lambda_n(G) = -\lambda_1(G)$. We define $\lambda(G) = \lambda_2(G)$.

\begin{definition}
Let $\Delta\geq 3$ be an integer and $\lambda$ be a positive real strictly less than $\Delta$. We let $\gbip$ denote the set of all connected $\Delta$-regular bipartite graphs $G$, for which $\lambda(G) \leq \lambda$. \qed
\end{definition}

One of the primary examples of bipartite expander graphs, and one of the main motivations behind this work, are random bipartite regular graphs. It is known \cite{brito2018spectral} that, for any fixed $\epsilon>0$, with high probability\footnote{Here and throughout the paper, ``with high probability'' means with probability tending to 1 as the size of the graph tends to infinity.} over the choice of a random bipartite $\Delta$-regular graph $G$, it holds that  $\lambda(G)\leq 2\sqrt{\Delta-1}+\epsilon$. From a counting/sampling perspective these graphs are particularly interesting since they have been key ingredients in obtaining inapproximability results \cite{sly2014,cai,antiferro}. Somewhat surprisingly, while we know constant factor\footnote{The constant factor estimates follow from applying the small subgraph conditioning method, which gives more accurate information about the partition function using the counts of small cycles, see \cite[Lemma 6.12]{antiferro}.} estimates of the partition function via (non-algorithmic) probabilistic methods that hold with probability $1-o(1)$ over the choice of the graph \cite{antiferro}, it is not known how to approximately sample from the Gibbs distribution efficiently. In fact,  even obtaining more refined estimates on the partition function is an open problem. As a corollary of our main result, we address both of these questions, provided that the degree $\Delta$ is sufficiently large relative to $H$.

To formally state our results, we will need some definitions. First,  the following computational problem that we will study. 

\prob{$H$, a symmetric matrix in  $\mathbb{R}^{q \times q}_{\geq 0}$, integer $\Delta\geq 3$, and a real $\lambda\in(0,\Delta)$.}
{$\spin_{H, \Delta, \lambda}$.}
	 {A graph $G \in \gbip$.}
 	 {The value of $Z_{G, H}$.}
In particular, we consider the problem of approximating $Z_{G, H}$ and sampling from $\mu_{G,H}$. Given an accuracy parameter $\epsilon > 0$, we say that $\hat{Z}$ is an $\epsilon$-approximation to $Z$ if $(1-\epsilon) Z \leq \hat{Z} \leq (1+\epsilon)Z$. For a distribution $\mu$, we say that a random variable $X$ is an $\epsilon$-sample from $\mu$ if the total variation distance between the distribution of $X$ and $\mu$ is at most $\epsilon$. A \textit{fully polynomial randomised approximation scheme} (FPRAS) for $\spin_{H, \Delta, \lambda}$ is a randomised algorithm that, given $\epsilon > 0$ and $G=(V,E) \in \gbip$ as input, outputs a random variable that is an $\epsilon$-approximation to $Z_{G, H}$ with probability at least $3/4$, in time $poly(|V|,1/\epsilon)$.\footnote{Note, the error probability can be reduced to any $\eta>0$ by calling the original $\mathsf{FPRAS}$ $O(\log(1/\eta))$ times.} 

We prove the following result. For a bipartite graph $G$, we use $(V_G^0, V_G^1)$ to denote the bipartition of the vertex set of $G$, and all logarithms throughout the paper are with base $e$.

\begin{restatable}{theorem}{ssfpras}\label{thm:ssfpras}
Let $q \geq 2$ be an integer,  $\delta$ be a real in $(0, 1)$, $H \in \mathbb{R}^{q \times q}_{\geq 0}$ be a symmetric $\delta$-matrix. Suppose that  $\Delta,\lambda$ satisfy $\tfrac{\Delta}{\lambda}\geq \tfrac{100}{1-\delta}q^2\log(q\Delta)$ and $\Delta \geq \big(\tfrac{10}{1-\delta}q\log(q\Delta)\big)^4$. Then,  there is an $\mathsf{FPRAS}$ for $\spin_{H, \Delta, \lambda}$. 

In fact, there is a randomised algorithm that, given a graph $G \in \gbip$ with $n=|V_G^0|=|V_G^1|$ vertices and an accuracy parameter $\epsilon^* \geq e^{-n/(5q)}$, outputs an $\epsilon^*$-approximation to $Z_{G,H}$ and an $\epsilon^*$-sample from the Gibbs distribution $\mu_{G,H}$ in time $O((n/\epsilon^*)^2 \log^4(n/\epsilon^*))$.
\end{restatable}

To prove Theorem~\ref{thm:ssfpras}, we first reduce the problem of approximating $Z_{G,H}$ to the problem of approximately sampling from the Gibbs distribution of a suitable polymer model.  In certain conditions, Chen et al.~\cite{chen2019fast} give an efficient Markov-chain based algorithm for approximately sampling from this Gibbs distribution. In Section~\ref{sec:final}, we show how to use this approximate sampling algorithm to give an FPRAS for $\spin_{H, \Delta, \lambda}$. Under similar conditions, it would also be possible to give an FPTAS for $\spin_{H, \Delta, \lambda}$, by appropriately truncating the polymers and applying the methods of~\cite{JKP}. However, the running time of algorithms obtained using the deterministic approach are slower, and typically of the form $n^{O(\log \Delta)}$. 

We remark that the restriction to bipartite expanders in Theorem~\ref{thm:ssfpras} is necessary to have a result that holds for general spin systems (which is our goal in this paper), see Section~\ref{sec:proofoutline} for a discussion of this point. 
As a corollary of Theorem~\ref{thm:ssfpras}, we have the following for random bipartite $\Delta$-regular graphs.

\begin{corollary}\label{cor:main}
Let $q \geq 2$ be an integer,  $\delta$ be a real in $(0, 1)$, and $H \in \mathbb{R}^{q \times q}_{\geq 0}$ be a symmetric $\delta$-matrix. Then, for all integers $\Delta\geq \big(\tfrac{10}{1-\delta}q\log(q\Delta)\big)^4$, there is a randomised algorithm such that the following holds with high probability over the choice of a random $\Delta$-regular bipartite graph $G$ with $n=|V^0_G|=|V^1_G|$. 

The algorithm, on input $G$ and an accuracy parameter $\epsilon^* \geq e^{-n/(5q)}$, outputs in time  $O((n/\epsilon^*)^2 \log^4(n/\epsilon^*))$ an $\epsilon^*$-approximation to the partition function $Z_{G,H}$ and an $\epsilon^*$-sample from the Gibbs distribution $\mu_{G,H}$. 
\end{corollary}
\begin{proof}
Using the result in \cite[Theorem 4]{brito2018spectral}, we have that, with high probability over the choice of $G$, it holds that $\lambda(G)\leq 2\sqrt{\Delta}$. It follows that $\tfrac{\Delta}{\lambda(G)}\geq \tfrac{1}{2}\sqrt{\Delta}$ and hence the result follows by applying Theorem~\ref{thm:ssfpras}.
\end{proof}

Our algorithms apply to a larger class of graphs when $\delta$ is small so that the interactions between spins are strong. By contrast, approaches such as MCMC and correlation decay apply when the interactions between spins are weak -- this corresponds to the so-called ``high-temperature'' regime, which is within the uniqueness phase of the infinite $\Delta$-regular tree. Since our results concern regular graphs, they easily extend to models with external fields -- the fields can be incorporated in the entries of the interaction matrix~$H$.

\subsection{Proof Outline}\label{sec:proofoutline}

In order to prove our main Theorem~\ref{thm:ssfpras}, we appeal to what are known as  
subset polymer models~\cite{KP, gruber1971general}.
Recently, polymer models have been used as a tool in the development of efficient counting and sampling algorithms~\cite{KP, helmuth2019algorithmic,  JKP,  liao2019counting, chen2019fast, cannon2020counting, Carlson} for problems that are not amenable to traditional approaches such as local-update Markov chains. Our approach is inspired by, and generalises, the approaches in \cite{JKP,chen2019fast}, where counting and sampling algorithms are given for the hard-core and ferromagnetic Potts models on expander graphs at low temperatures.

The main idea behind the use of polymer models is that, for graphs with good expansion properties, the partition function and the Gibbs distribution are dominated by configurations which are highly ordered, i.e., whose weight is large. As we shall see in detail in Section~\ref{sec:ground}, 
each of these large-weight configurations 
is close to a configuration that maps all vertices of~$G$ to a ``biclique'' of~$H$, as defined below.\footnote{The ``biclique'' terminology comes from the homomorphism problem (where $H$ corresponds to a graph), but our interpretation here is analogous which justifies its use.}

\begin{definition}\label{def:f334}
Let $q \geq 2$ be an integer, $\delta\in (0,1)$ be a real and $H \in \mathbb{R}_{\geq 0}^{q \times q}$ be a symmetric $\delta$-matrix. A \textit{biclique} of $H$ is a pair $(B_0, B_1)$ with $B_0, B_1 \subseteq [q]$, such that $H_{i, j} = 1$ for all $i \in B_0$ and all $j \in B_1$. We use $\calK_H$ to denote the set of all bicliques of $H$ and we use $\kmax$ to denote the set of all inclusion maximal bicliques of $H$.
\end{definition} 

Given a bipartite graph $G$, configurations $\sigma$ which assign vertices in $V_G^0$ a spin from $B_0$ and 
vertices in $V_G^1$ a spin from $B_1$ have weight 1. This is  the largest possible weight that any configuration can have, since $H$ is a $\delta$-matrix. 
Polymer models allow us to capture deviations from such configurations 
and to approximate their contribution to the partition function, see Section~\ref{sec:3f3frr}. Using the results of Sections~\ref{sec:ground} and~\ref{sec:3f3frr}, we give the proof of Theorem~\ref{thm:ssfpras} in Section~\ref{sec:final}.
 
 In some situations where
 Theorem~\ref{thm:ssfpras} 
 provides an FPRAS 
  for $\spin_{H, \Delta, \lambda}$
  it is easy to see that traditional approaches such as  Glauber dynamics
  do not give good approximation algorithms.
  In particular, when multiple ground states make non-negligible contributions to the partition function,
  these ground states provide a ``constriction in the state space'' which could also be used to prove that Glauber dynamics mixes slowly.
  Note, however, that there are many interaction matrices
  with a unique inclusion-maximal biclique and it is unclear whether Glauber dynamics would be slow.

 The inputs to $\spin_{H, \Delta, \lambda}$
 need to be bipartite graphs so that we can obtain an algorithm for all possible $H$. When the input $G$ is allowed to be non-bipartite, it is very unlikely that there is an algorithm that can work for general spin systems. To see this, consider the case where $H$ is the matrix $\left(\begin{smallmatrix}0 & 1\\ 1& 1\end{smallmatrix}\right)$. Recall
 from the remark following Definition~\ref{def:deltamatrix}
that this is a $\delta$-matrix for any $\delta \in (0,1)$. 
Spin configurations with this matrix~$H$ can be viewed as independent sets of~$G$ (where spin~$0$ means ``in the independent set'').
However, 
it is conjectured to be hard to find a maximum independent set on random $\Delta$-regular graphs for large $\Delta$, and the approximate 
counting problem that we study here is likely to be even harder.  
(This is in contrast to the bipartite case, where finding a maximum weight configuration is trivial since this is just a configuration assigning spins from a biclique.)

Despite the lack of a general result for non-bipartite graphs,
for certain spin systems,  our methods could be extended to allow non-bipartite inputs. For example, this occurs when all maximal bicliques $(B_0,B_1)$ of~$H$
have $B_0=B_1$ (so they can be viewed as ``cliques'' of~$H$). This is the case, for example, in the ferromagnetic Potts model, and it is the case more generally in ``ferromagnetic'' spin systems (where the interaction matrix has only positive eigenvalues, see, e.g., \cite{ferro}).

\section{Preliminaries}

Let $G$ be a bipartite graph. We will write $G$ as $(V_G^0, V_G^1, E_G)$, where $(V_G^0, V_G^1)$ denotes the bipartition of the vertex set of $G$ and $E_G$ its edge set; we will use $V_G= V_G^0\cup V_G^1$ to denote the  vertex set of $G$.

 For a vertex subset $S \subseteq V_G$, let $\partial_G S$ denote the set of vertices of $V_G \setminus S$ that have a neighbour in $S$, and let $S_G^+ = S \cup \partial_G S$. When $S=\{u\}$, we simply write $\partial_G u$ instead of $\partial_G S$. For vertex subsets $S, T \subseteq V_G$, let $E_G(S, T)$ denote the set of edges of $G$ that have one endpoint in $S$ and the other in $T$, and let $e_G(S, T) = |E_G(S, T)|$; when $S=T$, we simply write $E_G(S), e_G(S)$ instead of $E_G(S,S), e_G(S,S)$, respectively.  We will omit $G$ from all of the above notation where it is obvious from the context.

\subsection{Bipartite Expander Graphs}

 It is well-known that graphs in $\gbip$ have good expansion properties, and in this section we will review certain edge and vertex expansion properties that will be relevant for us.

The following result relates the spectrum of a regular bipartite graph to its edge-expansion properties. This result was first proven in~\cite[Theorem 5.1]{haemers1995interlacing}, though the version we state below is taken from~\cite{de2012large}.
\begin{lemma}[{\cite[Lemma 8]{de2012large}}]
\label{lem:emlemma}
Let $G = (V_G^0, V_G^1, E_G) \in \gbip$ with $n=|V_G^0|=|V_G^1|$. Then, for  sets $S_0 \subseteq V_G^0, S_1 \subseteq V_G^1$,  we have that
\[
\Big|e_G(S_0, S_1) - \frac{\Delta |S_0| |S_1|}{n} \Big| \leq \lambda \sqrt{|S_0| |S_1| \left( 1 - \frac{|S_0|}{n}\right) \left( 1 - \frac{|S_1|}{n}\right)}.
\]
\end{lemma}
The following simple consequence of the above result gives a lower bound on the edge expansion of $G \in \gbip$ when $\lambda$ is sufficiently small.
\begin{corollary}\label{lem:edgeexp}
Let $G = (V_G^0, V_G^1, E_G) \in \gbip$ with $n=|V_G^0|=|V_G^1|$. Then, for  sets $S_0 \subseteq V_G^0, S_1 \subseteq V_G^1$ such that $\lambda \leq \tfrac{\Delta}{2n}\sqrt{|S_0||S_1|}$, it holds that $e_G(S_0, S_1) \geq \tfrac{\Delta}{2n} |S_0| |S_1|$.
\end{corollary}
\begin{proof}
Lemma~\ref{lem:emlemma} implies that
\[\Big|e_G(S_0, S_1) - \frac{\Delta |S_0| |S_1|}{n} \Big| \leq \frac{\Delta |S_0| |S_1|}{2n}, \text{ therefore }  e_G(S_0, S_1) \geq \frac{\Delta |S_0| |S_1|}{2n}.
\qedhere\]
\end{proof}

A second combinatorial notion of expansion is vertex expansion. A well-known result from Tanner~\cite{tanner1984explicit} relates the spectrum of a graph to its vertex expansion properties (see also \cite{Kahale} for a more refined estimate). Here we state a version from~\cite[Theorem 4.15]{hoory2006expander}; there, the result is stated and proved for non-bipartite graphs, but a minor adaptation of the proof in \cite{hoory2006expander}, which we give for completeness here, also applies to bipartite graphs.
\begin{lemma}\label{lem:vtxexpansion}
Let $G = (V_G^0, V_G^1, E_G) \in \gbip$, $\rho>0$ be a real number and $i\in\{0,1\}$. Then,  for  all $S \subseteq V_G^i$ with $|S| \leq \rho |V_G^i|$, it holds that $|\partial S| \geq |S|/\big(\rho + \frac{\lambda^2}{\Delta^2}(1 - \rho)\big)$.
\end{lemma}
\begin{proof}
We prove the result for $i=0$; the case $i=1$ is symmetric. Let $S\subseteq V_G^0$ be such that $|S|\leq \rho n$ where $n=|V_G^0|=|V_G^1|$. We will show that $|\partial S| \geq |S|/(\rho + \frac{\lambda^2}{\Delta^2}(1 - \rho))$.

Let $A=A_G$ denote the adjacency matrix of $G$, and $B$ be the bi-adjacency matrix of $G$, so that $A=\big[\begin{smallmatrix} 0 & B\\ B^T& 0\end{smallmatrix}\big]$.
Then, we have that $\mu$ is an eigenvalue of $A$ iff $\mu^2$ is an eigenvalue of $B B^T$. 
Let  $v_1,\hdots,v_n$ be orthonormal eigenvectors of $B B^T$, with eigenvalues $\lambda_1^2,.., \lambda_n^2$, where $\lambda_1,\hdots, \lambda_n$ are the eigenvalues of $A$ in decreasing order. Note that  $\lambda_1^2 =\Delta^2$, $\lambda_2^2 =\lambda^2$ and since $G$ is $\Delta$-regular, we have that  $v_1=\frac{1}{\sqrt{n}}e_1$, where $e_1$ is the $n$-dimensional vector with all ones. 

From the spectral theorem, we have the decomposition  
\[BB^T= \sum_{i\in [n]} \lambda_i^2 v_i v_i^T.\]
Let $1_S$ be the $n$-dimensional vector whose $i$-th entry is equal to 1 if $i\in S$ and 0 otherwise. Since $v_1,\hdots,v_n$ is a basis of $\mathbb{R}^n$, we can write 
\begin{equation}\label{eq:gtgtgbt3}
1_S=\sum_{i\in [n]} \alpha_i v_i \mbox{ for some real numbers } \alpha_1,\hdots, \alpha_n.
\end{equation}
Using the orthonormality of $v_1,\hdots,v_n$, we obtain by multiplying \eqref{eq:gtgtgbt3} with $v_1^{T}$ that   $\alpha_1=|S|/\sqrt{n}\leq \rho\sqrt{n}$ and by considering the norm of $1_S$ that  $\sum_{i\in [n]}\alpha_i^2=|S|$. Hence, we have that 
\begin{equation}\label{eq:tgg5ygy5}
\left\|1_S^T B\right\|^2=  1_S^T B B^T 1_S= \sum_{i\in [n]} \lambda_i^2 \alpha_i^2\leq  \Delta^2\alpha_1^2+ \lambda^2(|S|-\alpha_1^2)\leq |S|\big(\Delta^2\rho+ \lambda ^2(1 -\rho)\big).
\end{equation}
For $j\in [n]$, let $u_j$ be the $j$-th entry in $1_S^T B$; observe that  $u_j$ is the number of neighbours  in $S$ of the $j$-th vertex in  $V_G^{1}$ and hence there are exactly $|\partial S|$ non-zero entries in $1_S^T B$. Moreover, since $G$ is $\Delta$-regular, we have that $\sum_{j\in \partial S}u_j=\sum_{j\in [n]}u_j=1_S^T B e_1=\Delta|S|$. From the Cauchy–Schwarz inequality, it follows that 
\begin{equation}\label{eq:tgg5ygy5b}
\left\|1_S^T B\right\|^2= \sum_{j\in [n]} u_j^2=\sum_{j\in \partial S} u_j^2 \geq  \frac{1}{|\partial S|}\Big(\sum_{j\in \partial S} u_j\Big)^2= \frac{1}{|\partial S|}\Big(\sum_{j\in [n]} u_j\Big)^2=\Delta^2|S|^2/|\partial S|.
\end{equation}  
Combining \eqref{eq:tgg5ygy5} and \eqref{eq:tgg5ygy5b} yields the desired inequality.
\end{proof}

\section{Ground states for spin configurations}\label{sec:ground}

In this section, we show that the partition function of a spin system is dominated by configurations which are ``close to maximal bicliques'', cf. Definition~\ref{def:f334}. Let $q\geq 2$, $\Delta\geq 3$ be integers and $\lambda, \delta$ be reals with $\lambda\in (0,\Delta)$ and $\delta\in (0,1)$. Let $G \in \gbip$ and let $H \in \mathbb{R}^{q \times q}_{\geq 0}$ be a symmetric $\delta$-matrix.  

We next describe more precisely the configurations  which are ``close'' to some maximal biclique of $H$. Given $\sigma: V_G\rightarrow [q]$ and a spin $i\in [q]$, we write $\sigma^{-1}(i)$ for the set of vertices of $G$ whose image  under $\sigma$ is $i$. More generally, for   a subset of spins $Q \subseteq [q]$, we  let $\sigma^{-1}(Q)=\{v\in V_G\mid \sigma(v)\in Q\}$.

\begin{definition}\label{def:rg5g5gge} 
Let $\epsilon \in (0, 1)$. For $(B_0, B_1) \in \kmax$, define $\Sigma_{G, H, \epsilon}^{B_0, B_1}$ to be the set of spin configurations $\sigma \in \Sigma_{G, H}$ for which 
\[\big|\sigma^{-1}(B_0) \cap V_G^0\big| + \big|\sigma^{-1}(B_1) \cap V_G^1\big| \geq (1 - \epsilon)|V_G|.\]
We define $\Sigma_{G, H, \epsilon}$ to be the union of the sets $\Sigma_{G, H, \epsilon}^{B_0, B_1}$ over all bicliques $(B_0, B_1)\in \kmax$ and define $Z_{G, H, \epsilon} :=  \sum_{\sigma \in \Sigma_{G, H, \epsilon}} w_{G, H}(\sigma)$.\qed
\end{definition}

The following result shows that $Z_{G, H, \epsilon}$ gives a close approximation to $Z_{G, H}$ whenever $\epsilon$ is sufficiently large relative to $\lambda,\Delta,q$.\footnote{Note, in Lemma~\ref{lem:gs}, as in other lemmas as well, our assumed inequalities for $\epsilon$ impose some restrictions on $\Delta,\lambda,q$ to ensure that such an $\epsilon$ exists. These restrictions will be carefully accounted for when we apply these lemmas; namely, in the proof of Theorem~\ref{thm:ssfpras}.}

\begin{lemma}
\label{lem:gs}
Let $\epsilon \in (0, 1)$  be such that  $\epsilon \geq 2q \lambda/\Delta$ and $\epsilon^2 \geq \frac{8q^2\log q}{\Delta \log(1/\delta)}$. Then, for $G \in \gbip$ with $n=|V^0_G|=|V^1_G|$, we have that $Z_{G, H, \epsilon}$ is an $e^{-n}$-approximation to $Z_{G, H}$.
\begin{proof}
We associate each spin configuration $\sigma \in \Sigma_{G, H}$ with a pair of spin subsets $\big(B_0(\sigma), B_1(\sigma)\big)$ by setting for $i\in \{0,1\}$
\[
B_i(\sigma) = \Big\{ j \in [q] : \big|\sigma^{-1}(j) \cap V_G^i\big| \geq \frac{\epsilon n}{q} \Big\}.
\]
Note that for $\sigma \in \Sigma_{G, H}$,  there are fewer than $\epsilon n$ vertices of $V_G^i$ which are \textit{not} assigned spins from $B_i(\sigma)$. Also note that, since $\epsilon\in (0,1)$ and $|V_G^i| = n$, we have that $B_i(\sigma) \neq \emptyset$ for $i \in \{0, 1\}$.

Fix arbitrary $\sigma \in \Sigma_{G, H}$. We first show  that 
\begin{equation}\label{eq:ssgrtgt12}
\mbox{ either $\big(B_0(\sigma), B_1(\sigma)\big) \in \calK_H$ or $w_{G, H}(\sigma) \leq \delta^{\Delta \epsilon^2 n/(2q^2)}$,}
\end{equation}
i.e., either $(B_0(\sigma), B_1(\sigma))$ is a biclique of $H$ or $\sigma$ has small weight. For $i \in \{0, 1\}$, consider arbitrary $j_i \in B_i(\sigma)$ and let $S_i = \sigma^{-1}(j_i) \cap V_G^i$. Since $|S_i| \geq \epsilon n/q$, we have that $\tfrac{\Delta}{2n} \sqrt{|S_0| |S_1|} \geq \Delta \epsilon/(2q) \geq \lambda$, thus it follows from Corollary~\ref{lem:edgeexp} that 
\[e(S_0, S_1) \geq \frac{\Delta |S_0| |S_1|}{2n} \geq \frac{\Delta \epsilon^2 n}{2q^2}.\]
It follows that, if $H_{j_0, j_1} \leq \delta$, then $w_{G, H}(\sigma) \leq \delta^{\Delta \epsilon^2 n/(2q^2)}$; otherwise, $H_{j_0, j_1} = 1$. Since $j_0,j_1$ were arbitrary spins in $B_0(\sigma),B_1(\sigma)$, respectively, we conclude \eqref{eq:ssgrtgt12}. 

Let $\sigma$ be such that $\big(B_0(\sigma), B_1(\sigma)\big) \in \calK_H$. Then there exists $(B_0, B_1) \in \kmax$ such that $B_i(\sigma) \subseteq B_i$ for $i \in \{0, 1\}$. Moreover, for $i \in \{0, 1\}$ and $j \in [q] \setminus B_i$, we have that $\big|\sigma^{-1}(j) \cap V_G^i\big| < \epsilon n/q$ and therefore that $|\sigma^{-1}([q] \setminus B_i) \cap V_G^i| < \epsilon n$. Hence, we conclude that 
\[\big|\sigma^{-1}(B_0) \cap V_G^0\big| + \big|\sigma^{-1}(B_1) \cap V_G^1\big| \geq (1 - \epsilon)|V_G|.\] 

Combining this with  \eqref{eq:ssgrtgt12}, we obtain that for all $\sigma \in \Sigma_{G, H} \setminus \Sigma_{G, H, \epsilon}$ it holds that   $w_{G, H}(\sigma) \leq \delta^{\Delta \epsilon^2 n/(2q^2)}$, and hence
\[
Z_{G, H} - Z_{G, H, \epsilon} \leq \sum_{\sigma \in \Sigma_{G, H} \setminus \Sigma_{G, H, \epsilon}} \delta^\frac{\Delta \epsilon^2 n}{2q^2} \leq q^{2n} \delta^\frac{\Delta \epsilon^2 n}{2q^2}\leq q^{-2n},
\]
where in the last inequality we used that $\Delta \geq \frac{8q^2\log(q)}{\epsilon^2 \log(1/\delta)}$. The result follows since $Z_{G, H} \geq 1$; this bound can be seen by considering the configuration that maps $V_G^0$ to $j_0$ and $V_G^1$ to $j_1$, where $j_0, j_1 \in [q]$ are such that $H_{j_0, j_1} = 1$.
\end{proof}
\end{lemma}

For our approximation algorithms, it will be useful to consider the following quantities.
\begin{definition}\label{def:zhateps}
For $\epsilon\in (0,1)$, let
\[\widehat{Z}_{G, H, \epsilon} = \sum_{(B_0, B_1) \in \kmax} \sum_{\sigma \in \Sigma_{G, H, \epsilon}^{B_0, B_1}} w_{G, H}(\sigma) \quad\mbox{ and }\quad Z_{G, H, \epsilon}^{\overlap}:=\sum_{\sigma \in \Sigma_{G, H, \epsilon}^{\overlap}}w_{G,H}(\sigma),\]
where $\Sigma_{G, H, \epsilon}^{\overlap}:=\bigcup_{(B_0,B_1)\in \kmax}\bigcup_{(C_0,C_1)\in \kmax\backslash \{(B_0,B_1)\}} \big(\Sigma^{B_0, B_1}_{G, H,\epsilon} \cap \Sigma^{C_0, C_1}_{G, H, \epsilon}\big)$.
\end{definition}

The following lemma shows, given a value of $\epsilon$ that is sufficiently small, that $\widehat{Z}_{G, H, 3\epsilon}$ is a close approximation to $Z_{G, H, \epsilon}$ by showing that $Z_{G, H, 3\epsilon}^{\overlap}$ is small relative to $Z_{G, H, \epsilon}$.
\begin{lemma}\label{lem:zhatepsappxheps}
Let  $\epsilon \in (0, \tfrac{1}{240q\log q}]$ be such that $\epsilon^2 \geq \frac{8q^2\log(q)}{\Delta \log(1/\delta)}$ and  $\epsilon \geq 2q \tfrac{\lambda}{\Delta}$. Then, for  all $G \in \gbip$ with $n=|V^0_G|=|V^1_G|$ sufficiently large, we have that $Z_{G, H, 3\epsilon}^{\overlap}\leq e^{-n/(3q)}Z_{G, H, \epsilon}$ and that $\widehat{Z}_{G, H, 3 \epsilon}$ is an $e^{-n/(4q)}$-approximation to $Z_{G, H, \epsilon}$.
\begin{proof}
We will show the result for all integers $n$ satisfying $n^2 2^{6q}e^{3n/(5q)}\leq e^{2n/(3q)}\big(1-\frac{2q}{n}\big)^2$. 

Since $\Sigma_{G, H, 3\epsilon}=\bigcup_{(B_0, B_1) \in \kmax}  \Sigma^{B_0, B_1}_{G, H, 3\epsilon}$ and $\Sigma_{G, H, \epsilon} \subseteq \Sigma_{G, H, 3\epsilon}$, we have 
\[\widehat{Z}_{G, H, 3\epsilon}\geq Z_{G, H, 3\epsilon}\geq Z_{G, H, \epsilon}.\] 
Moreover, each $\sigma \in \Sigma_{G, H, 3\epsilon}^{\overlap}$ is accounted for at most $|\kmax|\leq 2^{2q}$ times in $\widehat{Z}_{G, H, 3\epsilon}$, therefore
\[\widehat{Z}_{G, H, 3\epsilon} \leq 2^{2q}\widehat{Z}_{G, H, 3\epsilon}^{\overlap}+\sum_{\sigma \in \Sigma_{G, H, 3\epsilon}} w_{G, H}(\sigma).\]
Since $\Sigma_{G, H, \epsilon} \subseteq \Sigma_{G, H, 3\epsilon}$, we therefore have that
\begin{equation}
\label{eq:zhatzdiff}
\begin{aligned}
\widehat{Z}_{G, H, 3\epsilon} - Z_{G, H, \epsilon}& \leq 2^{2q}\widehat{Z}_{G, H, 3\epsilon}^{\overlap}+\sum_{\sigma \in \Sigma_{G, H, 3\epsilon} \setminus \Sigma_{G, H, \epsilon}} w_{G, H}(\sigma)\\
&\leq 2^{2q}\widehat{Z}_{G, H, 3\epsilon}^{\overlap}+Z_{G, H} - Z_{G, H, \epsilon}\leq 2^{2q}\widehat{Z}_{G, H, 3\epsilon}^{\overlap}+e^{-n} Z_{G, H, \epsilon},
\end{aligned}
\end{equation}
where the last inequality follows from Lemma~\ref{lem:gs}. The lemma will thus follow by showing 
\begin{equation}\label{eq:dblcount}
2^{2q}\widehat{Z}_{G, H, 3\epsilon}^{\overlap} \leq e^{-n/(3q)} Z_{G, H, \epsilon},
\end{equation}
since then, from~\eqref{eq:zhatzdiff} and~\eqref{eq:dblcount}, we obtain that  
\[\widehat{Z}_{G, H, 3\epsilon} - Z_{G, H, \epsilon} \leq \big(e^{-n} + e^{-n/(3q)}\big) \cdot Z_{G, H, \epsilon} \leq e^{-n/(4q)} \cdot Z_{G, H, \epsilon}.
\]

It remains to prove \eqref{eq:dblcount}. Consider $(B_0, B_1), (C_0, C_1) \in \kmax$ such that $(B_0, B_1) \neq (C_0, C_1)$ and let us upper bound the aggregate weight of  configurations in the set $F:=\Sigma^{B_0, B_1}_{G, H, 3\epsilon} \cap \Sigma^{C_0, C_1}_{G, H, 3\epsilon}$; in fact, we will just upper bound $|F|$ and use the trivial upper bound~1 on the weight of a configuration. For $\sigma\in F$, let  $S=\cup_{i\in \{0,1\}}(V_G^i\backslash \sigma^{-1}(B_i))$ and $T=\cup_{i\in \{0,1\}}(V_G^i\backslash \sigma^{-1}(C_i))$, so that  $|S|, |T| \leq 3\epsilon |V_G|= 6 \epsilon n$.  For $\epsilon\in (0,1/10)$, there are at most $\sum_{k = 0}^{\floor{6 \epsilon n}} \binom{2n}{k} \leq  n \binom{2n}{\floor{ 6 \epsilon n}}$ ways to choose each of $S$ and $T$, and then, crudely, $q^{12\epsilon n}$ ways to assign them spins; further, for $i\in \{0,1\}$, the vertices in $V_G^i\backslash (S\cup T)$ can be coloured in at most $|B_i\cap C_i|^n$ ways, since they must have a colour in $B_i\cap C_i$. Observe now that at least one of the inequalities $|B_0 \cap C_0| \leq  |B_0|-1,|B_1 \cap C_1| \leq  |B_1|-1$ must hold since otherwise $B_0\subseteq C_0, B_1\subseteq C_1$, contradicting the maximality $(B_0,B_1)\in \kmax$. We also have the bounds 
\[\Big(\frac{|B_i|-1}{|B_i|}\Big)^n\leq (1-1/q)^n\leq e^{-n/q}\mbox{ for $i\in \{0,1\}$.}\]
Combining the above, we obtain
\[
\sum_{\sigma \in \Sigma^{B_0, B_1}_{G, H, 3\epsilon} \cap \Sigma^{C_0, C_1}_{G, H, 3\epsilon}} w_{G, H}(\sigma)  \leq n^2\binom{2n}{\floor{6 \epsilon n}}^2 q^{12 \epsilon n} e^{-n/q}|B_0|^n|B_1|^n.
\]
We have $|\kmax| \leq 2^{2q}$ and $\binom{2n}{\floor{6 \epsilon n}} \leq (e/3\epsilon)^{6 \epsilon n}$, hence 
\begin{equation}\label{eq:556g56g6gyg}
\widehat{Z}_{G, H, 3\epsilon}^{\overlap}\leq n^2 2^{4q} (eq/3\epsilon)^{12 \epsilon n}e^{-n/q}\sum_{(B_0, B_1) \in \kmax} |B_0|^n |B_1|^n.
\end{equation}
For $c>0$, the function $f(x)=(c/x)^{x}$ is increasing in the interval $(0,c/e]$, so using that $\epsilon\leq \tfrac{1}{240q\log q}$, we further have that
\begin{equation}\label{eq:v5vt4f45f}
(eq/3\epsilon)^{\epsilon}\leq  (80eq^2\log q)^{1/(240q\log q)}=e^{\tfrac{1+\log(80)+2\log q+\log \log q}{240q\log q}} \leq e^{1/(20q)},
\end{equation}
where in the last inequality we used that $\tfrac{1+\log(80)}{240\log 2}+\tfrac{1}{120}+\frac{1}{240}\leq \frac{1}{20}$.
By considering the set of surjective maps from $V_G^0$ to $B_0$ and $V_G^1$ to $B_1$, for each maximal biclique $(B_0, B_1) \in \kmax$, we can also lower bound $Z_{G, H, \epsilon}$. Lemma 17 of~\cite{galanis2016approximately} (which is a simple corollary of an analogous bound in~\cite{DGGJ}) states that the number of surjective maps from a set of size $m$ to a set of size $k$ is at most $(1 - 2k/m) k^m$. Thus, we have that
\[Z_{G, H, \epsilon}\geq \Big(1-\frac{2q}{n}\Big)^2\sum_{(B_0, B_1) \in \kmax} |B_0|^n |B_1|^n.\] 
This, combined with \eqref{eq:556g56g6gyg}, \eqref{eq:v5vt4f45f} and the choice of $n$, yields \eqref{eq:dblcount}, finishing the proof of Lemma~\ref{lem:zhatepsappxheps}.
\end{proof}
\end{lemma}

\section{Using polymer models to estimate the partition function}\label{sec:3f3frr}
In this section, we first define subset polymer models, which will be important in obtaining our approximation algorithms, and then review the algorithmic results of \cite{chen2019fast} in Section~\ref{sec:algopol}. We then define a polymer model for spin systems in Section~\ref{sec:polymermodel}, and obtain approximation/sampling algorithms for it in Section~\ref{sec:sampling}.

\subsection{Subset polymer models}
\label{sec:abspolymodels}

 Our presentation of polymer models follows mostly~\cite{JKP}, but is slightly modified so that a polymer model is a function of both an underlying graph $G$ and a host graph $J_G$. This will be for convenience, since it allows algorithms that operate on polymer models to have access to both the underlying graph and the host graph (since it is possible that, in some polymer models, certain information about the structure of $G$ is lost when constructing $J_G$). 
 Our polymer models are subset polymers as defined by  Gruber and Kunz~\cite{gruber1971general}.

Let $\calG$ be a class of graphs. Given an underlying graph $G \in \calG$ and a set of spins $[q] = \{1,\hdots, q\}$, we construct a host graph $J_G$; in our case, we will set for example $J_G$ to be $G^3$ (see Section~\ref{sec:polymermodel} for more details), but in general other choices are obviously possible. We assign to each vertex $v \in V_{J_G}$ a set of ``ground state'' spins $g_v \subseteq [q]$. A \textit{polymer} is a pair $\gamma = (V_\gamma, \sigma_\gamma)$ consisting of a $J_G$-connected set of vertices $V_\gamma$ and an assignment $\sigma_\gamma : V_\gamma \rightarrow [q]$ such that $\sigma_\gamma(v) \in [q] \setminus g_v$ for all $v \in V_\gamma$. Let $\calP_{G}$ be the set of all polymers.

A \textit{polymer model} for an underlying graph $G$ and host graph $J_G$ is defined by a set of allowed polymers $\calC_G \subseteq \calP_G$, and a weight function $w_G : \calC_G \rightarrow \mathbb{R}_{\geq 0}$. For polymers $\gamma, \gamma' \in \calP_G$, we write $\gamma \sim \gamma'$ to denote that  $\gamma, \gamma'$ are \textit{compatible}, i.e.,  if $d_{J_G}(V_\gamma, V_{\gamma'}) > 1$ where $d_{J_G}(\cdot, \cdot)$ denotes the graph distance in $J_G$ and for $S, T \subseteq V_{J_G}$ we let $d_{J_G}(S, T) = \min_{u \in S, v \in T} d_{J_G}(u, v)$. We define $\Omega_G = \{ \Gamma \subseteq \calC_G \mid \forall \gamma, \gamma' \in \Gamma, \gamma \sim \gamma' \}$ to be the set of all sets of mutually compatible polymers of $\calC_G$; elements of $\Omega_G$ are called polymer configurations. The polymer model induces the partition function
\[
Z_{G} = \sum_{\Gamma \in \Omega_G} \prod_{\gamma \in \Gamma} w_G(\gamma),
\]
where $\prod_{\gamma \in \emptyset} w_G(\gamma) = 1$, and the following Gibbs distribution on $\Omega_G$, defined by 
\[
\mu_{G}(\Gamma) = \frac{\prod_{\gamma \in \Gamma} w_G(\gamma)}{Z_{G}}
\]
for all $\Gamma \in \Omega_G$. We use $(\calC_G, w_G ,J_G)$ to denote the polymer model and $\{(\calC_G, w_G ,J_G) \mid G \in \calG \}$ to denote the family of polymer models corresponding to the class of graphs $\calG$; we say that the family has degree bound $\Delta$ if, for every $G\in \calG$, both $G$ and the host graph $J_G$ have maximum degree at most $\Delta$. 

\subsection{Algorithms for polymer models}\label{sec:algopol}

Given a family of polymer models, Chen et. al. \cite{chen2019fast}, building upon work of \cite{JKP}, give sufficient conditions under which the partition function of the polymer model can be efficiently approximated using Markov chains. We will briefly describe these conditions, as well as the key results from~\cite{chen2019fast} that will be later important for us.

The first condition is known as \textit{computational feasibility} and is defined as follows.
\begin{definition}\cite[Definition 3]{chen2019fast}
Let $\calG$ be a class of graphs. A family of polymer models $\{(\calC_G, w_G, J_G) \mid G \in \calG \}$ is computationally feasible if, for all $G \in \calG$ and all $\gamma \in \calP_G$, we can determine whether $\gamma \in \calC_G$ and, if so, compute $w_G(\gamma)$ in time polynomial in $|V_\gamma|$.\qed
\end{definition}

The second condition is called the \textit{polymer sampling condition} and is defined as follows.
\begin{definition}\cite[Definition 4]{chen2019fast}
Let $q\geq 2,\Delta\geq 3$ be  integers, and  $\calG$ be a class of graphs. A family of polymer models $\{(\calC_G, w_G, J_G) \mid G \in \calG \}$  with $q$ spins and degree bound $\Delta$ satisfies the polymer sampling condition with constant $\tau \geq 5 + 3\log((q - 1)\Delta)$ if $w_G(\gamma) \leq e^{- \tau |V_\gamma|}$
for all $\gamma \in \calC_G$ and all $G \in \calG$.\qed
\end{definition}

The following result from \cite{chen2019fast} asserts that if a family of polymer models satisfies the above two conditions, then there are efficient approximation and sampling algorithms for the partition function and the Gibbs distribution of the polymer model, respectively.

\begin{theorem}[{\cite[Theorems 5 \& 6]{chen2019fast}}]
\label{thm:polymerfpras}
Let $q\geq 2$, $\Delta\geq 3$ be integers, and $\calG$ be a class of graphs. Suppose that $\{(\calC_G, w_G, J_G) \mid G \in \calG \}$ is a family of computationally feasible polymer models with $q$ spins and degree bound $\Delta$ that satisfies the polymer sampling condition.
Then there is a randomised algorithm that takes as input an $n$-vertex graph $G \in \calG$ and an accuracy parameter $\epsilon^* \in (0, 1)$ and  outputs an $\epsilon^*$-approximation to $Z_{G}$ in time $O((n/\epsilon^*)^2 \log^3(n/\epsilon^*))$ with probability at least~$3/4$. Moreover, there is a randomised algorithm that 
takes the same input and provides  an $\epsilon^*$-approximate sample from $\mu_G$ in time $O(n \log (n/\epsilon^*)\log(1/\epsilon^*))$. 
\end{theorem}

\subsection{Polymer model for spin systems}
\label{sec:polymermodel}

In this section we define a polymer model for spin systems that captures the deviations of spin configurations from maximal bicliques. The polymer model that we propose is a generalisation to arbitrary spin systems of a polymer model that was used in~\cite[Section 5]{JKP} in the case of proper colourings.

Let $H \in \mathbb{R}^{q \times q}_{\geq 0}$ be a symmetric matrix and $(B_0, B_1) \in \kmax$ be a maximal biclique of $H$. Let $G \in \gbip$ be a graph, and let $\epsilon \in (0, 1)$. The host graph for the polymer model is $J_G = G^3$, where $G^3$ is the graph defined on $V_G$ with two vertices connected by an edge if the distance between them in $G$ is at most $3$. Note that in the remainder of the paper, we will always use $\partial$ to denote the boundary with respect to $G$ of a vertex set, therefore we omit $G$ from this notation. For  $v \in V_G^i$ with $i\in\{0,1\}$, the set of ground state spins $g_v$ is $B_i$.  Let $\calP_{G, H}^{B_0, B_1}$ denote the set of all polymers, i.e., all pairs $\gamma = (V_\gamma, \sigma_\gamma)$ consisting of a $G^3$-connected set of vertices $V_\gamma$ and an assignment $\sigma_\gamma : V_\gamma \rightarrow [q]$ such that $\sigma_\gamma(v) \in [q] \setminus g_v$ for all $v \in V_\gamma$. We define the set of allowed polymers as 
\begin{equation}
\label{eq:polymers}
\calC_{G, H, \epsilon}^{B_0, B_1} = \left\{ \gamma\in  \calP_{G, H}^{B_0, B_1} : |V_\gamma|\leq \epsilon |V_G| \right\}
\end{equation}
and let $\Omega_{G, H, \epsilon}^{B_0, B_1}$ denote the set of all sets of mutually compatible polymers.  We define the weight of a polymer $\gamma=(V_\gamma,\sigma_\gamma) \in \calC_{G, H, \epsilon}^{B_0, B_1}$ as
\begin{equation}
\label{eq:polyweight}
w_{G, H}^{B_0, B_1}(\gamma)=\frac{\prod_{\{u,v\}\in E_G(V_\gamma)}H_{\sigma_\gamma(u), \sigma_\gamma(v)}\prod_{u \in \partial V_\gamma}F_u}{\prod_{i\in \{0,1\}}|B_i|^{| V_G^i\cap V_\gamma^+|}},
\end{equation}
where 
\[\mbox{for $u\in V_G^i$ with $i\in \{0,1\}$, } F_u:=\sum_{j \in B_i} \prod_{v \in V_\gamma\cap \partial u} H_{j, \sigma_\gamma(v)}.\]
We let $Z_{G, H, \epsilon}^{B_0, B_1}$ and $\mu_{G, H, \epsilon}^{B_0, B_1}$ denote the partition function and the Gibbs distribution of  the polymer model   $(\calC^{B_0, B_1}_{G, H, \epsilon}, w^{B_0, B_1}_{G, H}, J_G)$, as defined in Section~\ref{sec:abspolymodels}.

The next lemma shows the motivation behind the definition of the weight of a polymer. For a polymer configuration $\Gamma \in \Omega_{G, H, \epsilon}^{B_0, B_1}$, let $\cup \Gamma = \bigcup_{\gamma \in \Gamma} V_\gamma$, and $\sigma_\Gamma$ denote the assignment to vertices in $\cup \Gamma$ obtained by combining all of the assignments $\sigma_\gamma$, for $\gamma \in \Gamma$.

\begin{definition}\label{def:r544g5ga}
For $\Gamma \in \Omega_{G, H, \epsilon}^{B_0, B_1}$, define $\Sigma^{B_0, B_1}_{G, H}(\Gamma)$ to be the set of configurations $\tau$ such that $\tau|_{\cup \Gamma}=\sigma_\Gamma$ and which map, for $i\in \{0,1\}$,  $V_G^i \setminus (\cup \Gamma) $ to $B_i$.  \qed
\end{definition}

\begin{lemma}
\label{lem:weightmult}
Let $n=|V_G^0|=|V_G^1|$. For all $\epsilon \in (0, 1)$, and all polymer configurations $\Gamma \in \Omega_{G, H, \epsilon}^{B_0, B_1}$, we have that
\[
|B_0|^n |B_1|^n \prod_{\gamma \in \Gamma} w_{G, H}^{B_0, B_1}(\gamma) = \sum_{ \tau\in \Sigma^{B_0, B_1}_{G, H}( \Gamma)} w_{G, H}(\tau).
\]
\begin{proof}

Let $\gamma,\gamma'$ be two distinct polymers in $\Gamma$. Since the polymers are compatible, they correspond to distinct $G^3$-connected components; in particular, $V_{\gamma}^+\cap V_{\gamma'}^+=\emptyset$. It follows that $|(\cup \Gamma)^+|=\sum_{\gamma\in \Gamma} |V_\gamma^+|$. Hence, by the definition in~\eqref{eq:polyweight}, we have that
\begin{equation}\label{eq:prodwt}
\begin{aligned}
|B_0|^n |B_1|^n \prod_{\gamma \in \Gamma} w^{B_0,B_1}_{G,H}(\gamma) &= |B_0|^n |B_1|^n \prod_{\gamma \in \Gamma} \frac{\prod_{\{u,v\}\in E_G(V_\gamma)}H_{\sigma_\gamma(u), \sigma_\gamma(v)}\prod_{u \in \partial V_\gamma}F_u}{\prod_{i\in \{0,1\}}|B_i|^{| V_G^i\cap V_\gamma^+|}} \\
&= \prod_{i\in \{0,1\}}|B_i|^{|V^i_G\backslash (\cup \Gamma)^+|} \prod_{\{u,v\}\in E_G(\cup\Gamma)}H_{\sigma_\Gamma(u), \sigma_\Gamma(v)}\prod_{u \in \partial (\cup \Gamma)}F_u.
\end{aligned}
\end{equation}
On the other hand, for each $\tau\in \Sigma^{B_0, B_1}_{G, H}(\Gamma)$, we have that 
\[ w_{G, H}(\tau)=\prod_{\{u,v\}\in E_G(\cup \Gamma)}H_{\sigma_\Gamma(u), \sigma_\Gamma(v)}\prod_{u\in \partial (\cup \Gamma)}\prod_{v\in (\cup\Gamma)\cap \partial u}H_{\sigma_\Gamma(v), \tau(u)},\]
i.e., given $\Gamma$ and that $\tau\in \Sigma^{B_0, B_1}_{G, H}(\Gamma)$, the weight of $\tau$ depends only on the assignment of $\partial (\cup \Gamma)$ --- this is because $H_{\tau(u), \tau(v)} = 1$ for all edges $\{u, v\}$ with $u, v \notin \cup \Gamma$. Let $W(\Gamma)$ be the set of configurations $\eta: \partial (\cup \Gamma)\rightarrow B_0\cup B_1$ such that vertices in $\partial (\cup \Gamma)\cap V_G^i$ take a spin in $B_i$, for $i\in \{0,1\}$. For $\eta\in W(\Gamma)$, the number of configurations $\tau$ in $\Sigma^{B_0, B_1}_{G, H}(\Gamma)$ with $\tau|_{\partial(\cup \Gamma)}=\eta$ is $\prod_{i\in \{0,1\}}|B_i|^{|V^i_G\backslash (\cup \Gamma)^+|} $, and 
\[\sum_{\eta\in W(\Gamma)}\prod_{u\in \partial (\cup \Gamma)}\prod_{v\in (\cup\Gamma)\cap \partial u }H_{\sigma_\Gamma(v), \eta(u)}=\prod_{i\in \{0,1\}}\prod_{u\in V_G^i\cap  \partial (\cup \Gamma)}\sum_{\eta:\{u\}\rightarrow B_i}\prod_{v\in (\cup\Gamma)\cap \partial u}H_{\sigma_\Gamma(v), \eta(u)}=\prod_{u \in \partial (\cup \Gamma)}F_u.\] 
It follows that  
\[\sum_{ \tau\in \Sigma^{B_0, B_1}_{G, H}(\Gamma)} w_{G, H}(\tau)=\prod_{i\in \{0,1\}}|B_i|^{|V^i_G\backslash (\cup \Gamma)^+|} \prod_{\{u,v\}\in E_G(V_\Gamma)}H_{\sigma_\Gamma(u), \sigma_\Gamma(v)}\prod_{u \in \partial (\cup \Gamma)}F_u,\]
which combined with \eqref{eq:prodwt} gives the desired equality.
\end{proof}
\end{lemma}
The following quantity combines the partition functions of  the polymer models corresponding to maximal bicliques of $H$; we will use this as our approximation to $Z_{G,H}$.
\begin{definition}\label{def:Zpolymer}
For $\epsilon\in (0,1)$, let
\[
Z_{G, H, \epsilon}^{\text{polymer}} = \sum_{(B_0, B_1) \in \kmax} |B_0|^n |B_1|^n \cdot  Z_{G, H, \epsilon}^{B_0, B_1}.\qedhere
\]
\end{definition}
The following lemma is a minor adaptation of~\cite[Claim 29]{JKP} in our setting, and will be used to bound the aggregate size of polymer configurations.
\begin{lemma}[{\cite[Claim 29]{JKP}}]
\label{lem:smallconfigs}
Let $\epsilon \in (0, 1)$ be such that $\epsilon\geq \lambda/\Delta$. Then, for $G \in \gbip$ with $n=|V^0_G|=|V^1_G|$ sufficiently large, there is no set  $S \subseteq V_G$ with  $|S|>6\epsilon n$ whose $G^3$-connected components, say $S_1, S_2,\hdots, S_k$, satisfy $|S_i| \leq 2\epsilon n$ for $i\in [k]$.
\begin{proof}
\sloppy We will show the result when  $n\geq 1/(2\epsilon^2 \Delta)$. For the sake of contradiction, suppose that such a set $S$ exists with $|S|>6\epsilon n$ and whose $G^3$-connected components  $S_1, S_2,\hdots, S_k$ satisfy $|S_i| \leq 2\epsilon n$ for $i\in [k]$.  Then, we can partition $\{1, 2,\hdots, k\}$ into sets $T_1$ and $T_2$ such that, for $j\in \{1,2\}$, the set $U_j=\cup_{i\in T_j}S_i$ satisfies $|U_j|\geq 2 \epsilon n$. It follows by Corollary~\ref{lem:edgeexp} that $e_G(U_1, U_2) \geq 2\epsilon^2 \Delta n \geq 1$, therefore there must exist $i \in T_1$ and $j \in T_2$ such that $e_G(S_i, S_j) \geq 1$.  This yields a contradiction since two distinct $G^3$-connected components cannot be connected by an edge.
\end{proof}
\end{lemma}

Finally, we have the following result.
\begin{lemma}
\label{lem:pmerpfappx}
Let  $\epsilon \in (0, \tfrac{1}{240q\log q}]$ be such that $\Delta \geq \frac{8q^2\log(q)}{\epsilon^2 \log(1/\delta)}$ and  $\epsilon\geq 2q\tfrac{\lambda}{\Delta}$. Then, for  all $G \in \gbip$ with $n=|V^0_G|=|V^1_G|$ sufficiently large, we have that $Z_{G, H, \epsilon}^{\text{polymer}}$ is an $e^{-n/(4q)}$-approximation to $Z_{G, H, \epsilon}$.
\begin{proof}
Using the definition of the polymer partition function $Z_{G,H,\epsilon}^{B_0,B_1}$ (see Section~\ref{sec:polymermodel}) and Lemma~\ref{lem:weightmult}, we can rewrite $Z_{G, H, \epsilon}^{\text{polymer}}$ from Definition~\ref{def:Zpolymer} as
\begin{equation}\label{eq:zpolymerdbl}
\begin{aligned}
Z_{G, H, \epsilon}^{\text{polymer}}& =  
\sum_{(B_0,B_1)\in \kmax} \sum_{\Gamma \in \Omega_{G,H,\epsilon}^{B_0,B_1}} |B_0|^n |B_1|^n\prod_{\gamma\in \Gamma} w_{G,H}^{B_0,B_1}(\gamma)\\
&=\sum_{(B_0, B_1) \in \kmax} \sum_{\Gamma \in \Omega_{G, H, \epsilon}^{B_0, B_1}} \sum_{\sigma \in \Sigma^{B_0, B_1}_{G, H}( \Gamma)} w_{G, H}(\sigma).
\end{aligned}
\end{equation}
Recall from Definition~\ref{def:r544g5ga} that $\Sigma_{G,H}^{B_0,B_1}(\Gamma)$ is
the set of configurations $\tau$ such that $\tau|_{\cup \Gamma}= \sigma_\Gamma$ and which map, for
$i \in \{0,1\}$, $V_G^i\backslash(\cup \Gamma)$ to $B_i$. Recall also from Definition~\ref{def:rg5g5gge} that $\Sigma_{G, H, \epsilon}^{B_0, B_1}$ is the set of $\sigma \in \Sigma_{G, H}$ for which $\big|\sigma^{-1}(B_0) \cap V_G^0\big| + \big|\sigma^{-1}(B_1) \cap V_G^1\big| \geq (1 - \epsilon)|V_G|$ and that $Z_{G, H, \epsilon} :=  \sum_{\sigma \in \Sigma_{G, H, \epsilon}} w_{G, H}(\sigma)$.
 
We first prove that $Z_{G, H, \epsilon}^{\text{polymer}} \geq Z_{G, H, \epsilon}$. Index the  bicliques in $\kmax$ arbitrarily. For $\tau \in \Sigma_{G, H, \epsilon}$, let $\big(B_{0,\tau},B_{1,\tau}\big)$ be the biclique $(B_0,B_1)\in \kmax$ with the smallest index such that   $\tau \in \Sigma^{B_0, B_1}_{G, H,\epsilon}$. Let $T_i=V_G^i\cap  \tau^{-1}([q]\backslash B_{i,\tau})$ for $i\in \{0,1\}$, so that $|T_0\cup T_1|\leq \epsilon |V_G|$. Let $S_1,\hdots, S_k$ denote the $G^3$-connected components of $T_0\cup T_1$. Consider the polymer configuration $\Gamma_\tau$ which is the union of the polymers $(S_1,\tau|_{S_1}),\hdots, (S_k,\tau|_{S_k})$. Note that the map $\tau\mapsto (B_{0,\tau}, B_{1,\tau}, \Gamma_\tau)$ is injective. Moreover, we have that $\Gamma_\tau\in \Omega_{G, H, \epsilon}^{B_{0,\tau}, B_{1,\tau}}$ and $\tau\in \Sigma^{B_{0,\tau}, B_{1,\tau}}_{G, H}(\Gamma_\tau)$, so from \eqref{eq:zpolymerdbl} we obtain that $Z_{G, H, \epsilon}^{\text{polymer}} \geq Z_{G, H, \epsilon}$.

We next prove that $Z_{G, H, \epsilon}^{\text{polymer}} \leq (1 + e^{-n/(4q)}) Z_{G, H, \epsilon}$. Recall from Definition~\ref{def:zhateps} that 
\begin{equation*}
\widehat{Z}_{G, H, 3\epsilon} = \sum_{(B_0, B_1) \in \kmax} \sum_{\sigma \in \Sigma_{G, H, 3\epsilon}^{B_0, B_1}} w_{G, H}(\sigma).
\end{equation*}
Fix any $(B_0,B_1)\in \kmax$. For $\Gamma \in \Omega_{G, H, \epsilon}^{B_0, B_1}$, we have from Lemma~\ref{lem:smallconfigs} that $|\cup \Gamma|\leq 6\epsilon n= 3\epsilon |V_G|$. Therefore, for $\sigma \in \Sigma^{B_0, B_1}_{G, H}(\Gamma)$ as in Definition~\ref{def:r544g5ga}, we have that $\big|\sigma^{-1}(B_0) \cap V_G^0\big| + \big|\sigma^{-1}(B_1) \cap V_G^1\big| \geq (1 - 3\epsilon)|V_G|$ and hence $\sigma\in \Sigma^{B_0, B_1}_{G, H,3\epsilon}$ (cf. Definition~\ref{def:rg5g5gge}). Note that the map $(\Gamma,\sigma)\mapsto \sigma$ is injective (since $\Gamma$ can be recovered from $\sigma$ using the biclique $(B_0,B_1)$), yielding that
\begin{equation*}
\sum_{\Gamma \in \Omega_{G, H, \epsilon}^{B_0, B_1}} \sum_{\sigma \in \Sigma^{B_0, B_1}_{G, H}(\Gamma)} w_{G, H}(\sigma)\leq \sum_{\sigma \in \Sigma_{G, H, 3\epsilon}^{B_0, B_1}} w_{G, H}(\sigma),
\end{equation*}
and hence $Z_{G, H, \epsilon}^{\text{polymer}} \leq \widehat{Z}_{G, H, 3 \epsilon}$. By Lemma~\ref{lem:zhatepsappxheps}, we obtain that $Z_{G, H, \epsilon}^{\text{polymer}} \leq (1 + e^{-n/(4q)}) Z_{G, H, \epsilon}$.

This finishes the proof of Lemma~\ref{lem:pmerpfappx}.
\end{proof}
\end{lemma}

\subsection{Sampling from the polymer model}
\label{sec:sampling}

Let  $q\geq 2,\Delta\geq 3$ be integers and $\delta\in (0,1), \lambda\in (0,\Delta)$ be reals. Let $H \in \mathbb{R}^{q \times q}_{\geq 0}$ be a symmetric $\delta$-matrix and $(B_0, B_1) \in \kmax$. For $\epsilon \in (0, 1)$, we now show that the family of polymer models $\{ (\calC^{B_0, B_1}_{G, H, \epsilon}, w^{B_0, B_1}_{G, H}, J_G) \mid G \in \gbip \}$ which was defined in the previous subsection is computationally feasible and satisfies the polymer sampling condition. 

\begin{lemma}\label{lem:verify}
Let $\epsilon \in (0, 1)$ be such that $\epsilon \geq \lambda^2/\Delta^2$ and $\epsilon \leq \frac{1-\delta}{40q\log (q \Delta)}$. The family of polymer models $\{ (\calC^{B_0, B_1}_{G, H, \epsilon}, w^{B_0, B_1}_{G, H}, J_G) \mid G \in \gbip \}$ is computationally feasible and satisfies the polymer sampling condition with constant $\tau \geq 5 + 3\log((q - 1)\Delta^3)$.
\begin{proof}
Consider an arbitrary polymer $\gamma=(V_\gamma,\sigma_\gamma) \in \calP_{G, H}^{B_0, B_1}$.

We  first verify computational feasibility by showing that we can determine whether $\gamma \in \calC_{G, H, \epsilon}^{B_0, B_1}$ and, if so, compute its weight $w_{G, H}^{B_0, B_1}(\gamma)$ in time $O(|V_\gamma|)$. Indeed, we have that  $\gamma \in \calC_{G, H, \epsilon}^{B_0, B_1}$ iff $|V_\gamma| \leq 2 \epsilon n$ and $\sigma_\gamma$ maps $V_\gamma \cap V_G^i$ to $[q] \setminus B_i$ for each $i \in \{0, 1\}$, which can be clearly checked in time $O(|V_\gamma|)$. To compute  $w_{G, H}^{B_0, B_1}(\gamma)$, recall from~\eqref{eq:polyweight} that 
\begin{equation*}\tag{\ref{eq:polyweight}}
w_{G, H}^{B_0, B_1}(\gamma)=\frac{\prod_{\{u,v\}\in E_G(V_\gamma)}H_{\sigma_\gamma(u), \sigma_\gamma(v)}\prod_{u \in \partial V_\gamma}F_u}{\prod_{i\in \{0,1\}}|B_i|^{| V_G^i\cap V_\gamma^+|}},
\end{equation*}
where for a vertex $u\in V_G^i \cap \partial V_\gamma$ with $i\in \{0,1\}$, $F_u:=\sum_{j \in B_i} \prod_{v \in \partial u \cap V_\gamma} H_{j, \sigma_\gamma(v)}$. It follows that the r.h.s. in \eqref{eq:polyweight} can be computed in $O(|V_\gamma|)$ time, using that $|\partial V_\gamma| \leq \Delta |V_\gamma|$.

We next verify the polymer sampling condition. From \eqref{eq:polyweight}, using that the entries of $H$ are at most~1, we have the bound
\begin{equation}\label{eq:4tf4ftf4ftf}
w_{G, H}^{B_0, B_1}(\gamma)\leq \frac{\prod_{u \in \partial V_\gamma \cap V_G}F_u}{\prod_{i\in \{0,1\}}|B_i|^{| V_G^i\cap V_\gamma^+|}}.
\end{equation}
Let us now consider the factor $F_u$ for $u\in \partial V_\gamma \cap V_G^i$ with $i \in \{0, 1\}$. Let $v$ be a neighbour of $u$ in $V_\gamma \cap V_G^{i\oplus 1}$; such $v$ exists since $u\in \partial V_\gamma \cap V_G^i$. Then, there exists $j\in B_i$ such that $H_{j, \sigma_\gamma(v)}\leq \delta$; otherwise $(B_i, B_{i\oplus1}\cup \{\sigma_\gamma(v)\})$ would also be a biclique of $H$, contradicting the maximality of $(B_0,B_1)$ (since $\sigma_\gamma(v)\notin B_{i\oplus 1}$). It follows that $F_u\leq |B_i|-1+\delta$. Using this in \eqref{eq:4tf4ftf4ftf}, we get that
\begin{equation}\label{eq:rgbtbbtyrtrr355}
w_{G, H}^{B_0, B_1}(\gamma)\leq \frac{\prod_{i\in\{0,1\}}(|B_i| - 1 + \delta)^{|\partial V_\gamma \cap V_G^i|}}{\prod_{i\in \{0,1\}}|B_i|^{| V_G^i\cap V_\gamma^+|}}\leq \Big( 1 - \frac{1 - \delta}{q} \Big)^{|\partial V_\gamma|}\leq e^{-|\partial V_\gamma| \big( \tfrac{1 - \delta}{ q} \big)},
\end{equation}
where in the second to last inequality we used that $|V_\gamma^+\cap V_G^i|\geq |\partial V_\gamma \cap V_G^i|$ for $i\in\{0,1\}$. 

We next lower bound $|\partial V_\gamma|$ in terms of $|V_\gamma|$. For $i\in \{0,1\}$, let $\rho_i=|V_\gamma\cap V_G^i|/n$ and $\rho=|V_\gamma|/n$ . Applying Lemma~\ref{lem:vtxexpansion} to the set $V_\gamma\cap V_G^i$, we have that
\[
|\partial (V_\gamma\cap V_G^i)| \geq \frac{|V_\gamma\cap V_G^i|}{\rho_i + \tfrac{\lambda^2}{\Delta^2}(1 - \rho_i)} \geq \frac{|V_\gamma\cap V_G^i|}{\rho_i+\tfrac{\lambda^2}{\Delta^2}}.
\]
Now observe that \[\sum_{i\in\{0,1\}}\rho_i |V_\gamma\cap V_G^i|=\tfrac{1}{n}\sum_{i\in\{0,1\}} |V_\gamma\cap V_G^i|^2\leq \tfrac{1}{n}|V_\gamma|^2 =\rho |V_\gamma|\]
and hence, using the inequality $\tfrac{a}{x}+\tfrac{b}{y}\geq \tfrac{(a+b)^2}{ax+by}$ which holds for all $a,b,x,y\geq 0$, we obtain
\[ \sum_{i\in \{0,1\}}|\partial (V_\gamma\cap V_G^i)|\geq \sum_{i\in \{0,1\}}\frac{|V_\gamma\cap V_G^i|}{\rho_i+\tfrac{\lambda^2}{\Delta^2}}\geq \frac{|V_\gamma|}{\rho+\tfrac{\lambda^2}{\Delta^2}}\geq \frac{|V_\gamma|}{3\epsilon},\]
using that $\rho\leq 2\epsilon$ and $\epsilon \geq \tfrac{\lambda^2}{\Delta^2}$. Since $\epsilon\in (0,1/20)$, we have that $|\partial V_\gamma|\geq \frac{|V_\gamma|}{3\epsilon}-|V_\gamma|\geq \frac{|V_\gamma|}{4\epsilon}$. Plugging this into \eqref{eq:rgbtbbtyrtrr355}, we obtain
\[
w_{G, H}^{B_0, B_1}(\gamma) \leq e^{-| V_\gamma| \big( \tfrac{1 - \delta}{4 \epsilon q} \big)} \leq e^{-\tau |V_\gamma|},
\]
where $\tau :=\frac{1-\delta}{4\epsilon q}\geq 10\log (q \Delta)\geq 5 + 3\log((q-1)\Delta^3)$ using that $\epsilon \leq \frac{1-\delta}{40q\log (q \Delta)}$ and  $q\geq 2, \Delta\geq 3$. This finishes the proof of Lemma~\ref{lem:verify}, after observing that the degree bound for the family $\{ (\calC^{B_0, B_1}_{G, H, \epsilon}, w^{B_0, B_1}_{G, H}, J_G) \mid G \in \gbip \}$ is $\Delta^3$.  
\end{proof}
\end{lemma}

Finally, we can apply Theorem~\ref{thm:polymerfpras}, which gives us an efficient algorithm for approximating the partition function of the polymer model.

\begin{corollary}
\label{cor:pmerfpras}
Let $q \geq 2$ and $\delta \in (0, 1)$. Let $H\in \mathbb{R}^{q \times q}_{\geq 0}$ be a symmetric $\delta$-matrix and let $(B_0, B_1) \in \kmax$. Let $\epsilon \in (0, 1)$ be such that $\epsilon \geq \lambda^2/\Delta^2$ and $\epsilon \leq \frac{1-\delta}{40q\log (q \Delta)}$. 
Then there is a randomised algorithm that takes as input an $n$-vertex graph $G \in \gbip$ and an accuracy parameter $\epsilon^* \in (0, 1)$ and  outputs an $\epsilon^*$-approximation to $Z_{G,H,\epsilon}^{B_0,B_1}$ in time $O((n/\epsilon^*)^2 \log^3(n/\epsilon^*))$ with probability at least~$3/4$. Moreover, there is a randomised algorithm that 
takes the same input and  outputs  an $\epsilon^*$-approximate sample from $\mu_{G,H,\epsilon}^{B_0,B_1}$ in time $O(n \log (n/\epsilon^*)\log(1/\epsilon^*))$. 
\end{corollary} 

\section{Proof of Theorem~\ref{thm:ssfpras}}\label{sec:final}

In this section, we combine the results of Sections~\ref{sec:ground} and~\ref{sec:3f3frr} to prove Theorem~\ref{thm:ssfpras}.

\ssfpras*
\begin{proof}
As input to the FPRAS, we are given a graph $G \in \gbip$ and an accuracy parameter $\epsilon^* \in (0, 1)$. We may assume that $n=|V^0_G|=|V^1_G|$ is sufficiently large, otherwise we can compute $Z_{G, H}$ exactly in constant time, by brute force. Similarly, we may assume that $\epsilon^* \geq 9 e^{-n/(4q)}$, otherwise  we can compute $Z_{G, H}$ exactly in $O(n q^{2n}) = \text{poly}(1/\epsilon^*)$ time, by brute force. 

Let $\epsilon = \tfrac{1-\delta}{50q\log(q\Delta)}$ and observe that, using the lower bounds on $\tfrac{\Delta}{\lambda}$ and $\Delta$, we have that
\[\epsilon\leq \frac{1}{240q\log q}, \quad \epsilon \geq 2q \frac{\lambda}{\Delta}\geq \frac{\lambda^2}{\Delta^2}, \quad\epsilon^2 \geq \frac{8q^2\log q}{\Delta \log(1/\delta)},\]
where the first inequality follows from $\Delta\geq q^4$, the second from rearranging  $\tfrac{\Delta}{\lambda}\geq 2q/\epsilon$ (using the lower bound on $\tfrac{\Delta}{\lambda}$), and the last inequality from using that $\log(1/\delta)\geq 1-\delta$ for all $\delta\in (0,1)$. In particular, the assumption of Lemmas~\ref{lem:gs}, \ref{lem:pmerpfappx} and Corollary~\ref{cor:pmerfpras} are satisfied.

By Corollary~\ref{cor:pmerfpras}, for an arbitrary biclique $(B_0,B_1)\in \kmax$, we  can obtain an $(\epsilon^*/8)$-approximation to $Z^{B_0, B_1}_{G, H, \epsilon}$ in $O((n/\epsilon^*)^2 \log^3(n/\epsilon^*))$ time, with probability at least $3/4$. Taking the median of $O(\log (1/\epsilon^*))$ runs of this algorithm, we therefore obtain $\widehat{Z}^{B_0, B_1}_{G, H, \epsilon}$ which is an $(\epsilon^*/8)$-approximation to $Z^{B_0, B_1}_{G, H, \epsilon}$ with probability at least $1 - \epsilon^*/(16|\kmax|)$. By a union bound over the bicliques in $\kmax$, it follows that
\[
\widehat{Z}_{G, H, \epsilon}^{\text{polymer}} := \sum_{(B_0, B_1) \in \kmax} |B_0|^n |B_1|^n \cdot  \widehat{Z}^{B_0, B_1}_{G, H, \epsilon},
\]
is an $(\epsilon^*/8)$-approximation to $Z_{G, H, \epsilon}^{\text{polymer}}$ (cf. Definition~\ref{def:Zpolymer}) with probability at least $1-\epsilon^*/16$. By Lemma~\ref{lem:pmerpfappx}, $Z_{G, H, \epsilon}^{\text{polymer}}$ is an $(\epsilon^*/8)$-approximation to $Z_{G, H, \epsilon}$ which, by Lemma~\ref{lem:gs}, is an $(\epsilon^*/8)$-approximation to $Z_{G, H}$. It therefore follows that $\widehat{Z}_{G, H, \epsilon}^{\text{polymer}}$ is a $(3\epsilon^*/8)$-approximation, and hence an $\epsilon^*$-approximation, to $Z_{G, H}$ with probability at least $1-\epsilon^*/16$. The total run-time of the algorithm is therefore $O((n/\epsilon^*)^2 \log^4(n/\epsilon^*))$.

For the sampling algorithm, we assume again that $\epsilon^* \geq  9e^{-n/(4q)}$.  We  first sample a biclique $\widehat{\Bb}=(B_0,B_1)\in \kmax$ with probability $\tfrac{|B_0|^n |B_1|^n\widehat{Z}^{B_0, B_1}_{G, H, \epsilon}}{\widehat{Z}_{G, H, \epsilon}^{\text{polymer}}}$, where $\widehat{Z}^{B_0, B_1}_{G, H, \epsilon}$  and $\widehat{Z}_{G, H, \epsilon}^{\text{polymer}}$ are as before. Then, using Corollary~\ref{cor:pmerfpras}, we sample a polymer configuration $\widehat{\Gammab}$ whose distribution is at distance at most $\epsilon^*/6$ from $\mu^{B_0,B_1}_{G,H,\epsilon}$.  We output $\hat{\sigmab}=\texttt{Spin}_{\widehat{\Bb}}(\widehat{\Gammab})$, where for a biclique $B=(B_0,B_1)$ and a polymer configuration $\Gamma$, $\texttt{Spin}_{B}(\Gamma)$ is a random configuration $\tau$ obtained  as follows:
\begin{itemize}
\item For every vertex $u\in \cup \Gamma$, we set $\tau(u)=\sigma_\Gamma(u)$.
\item For $u\in V_G^i\backslash (\cup \Gamma)^+$ with $i\in \{0,1\}$, we assign a random spin from $B_i$ uniformly at random.
\item For $u \in \partial(\cup \Gamma) \cap V_G^i$ with $i\in \{0,1\}$, for $j\in B_i$ we set $\tau(u)=j$ with probability $\frac{1}{F_u}\prod_{v \in \partial u \cap (\cup \Gamma)} H_{j, \sigma_\Gamma(v)}$ where $F_u:=\sum_{j \in B_i} \prod_{v \in \partial u  \cap (\cup \Gamma)} H_{j, \sigma_{\Gamma}(v)}$.
\end{itemize}
We claim that $\hat{\sigmab}$ is an $\epsilon^*$-approximate sample from the Gibbs distribution $\mu_{G,H}$. To prove this, let $\Bb,\Gammab, \sigmab$ be the analogues of $\wBb,\widehat{\Gammab}, \hat{\sigmab}$, respectively, when there is no error, more precisely:
\begin{enumerate}
\item $\Bb=(B_0,B_1)$ with probability $|B_0|^n |B_1|^n\tfrac{ Z^{B_0, B_1}_{G, H, \epsilon}}{Z_{G, H, \epsilon}^{\text{polymer}}}$.
\item Conditioned on  $\Bb=(B_0,B_1)$,  $\Gammab\sim \mu^{B_0,B_1}_{G,H,\epsilon}$ and $\sigmab=\texttt{Spin}_{\Bb}(\Gammab)$. 
\end{enumerate} 
With probability $1-\epsilon^*/16$, we have that, for all bicliques $(B_0,B_1)\in \kmax$, $\widehat{Z}^{B_0, B_1}_{G, H, \epsilon}$  and $\widehat{Z}_{G, H, \epsilon}^{\text{polymer}}$ are $(\epsilon^*/8)$-approximations to $\widehat{Z}^{B_0, B_1}_{G, H, \epsilon}$  and $\widehat{Z}_{G, H, \epsilon}^{\text{polymer}}$ respectively, and conditioned on this the  total variation distance between the distributions of $\widehat{\Bb}$ and $\Bb$ is at most $\frac{1+\epsilon^*/8}{1-\epsilon^*/8}-1\leq 3\epsilon^*/8$. It follows that the total variation distance between the distributions of $\widehat{\Bb}$ and $\Bb$ is at most $3\epsilon^*/8+\epsilon^*/16\leq\epsilon^*/2$, and so there is a coupling so that $\Pr[\widehat{\Bb}\neq \Bb]\leq \epsilon^*/2$. Conditioned on $\widehat{\Bb}=\Bb$, we have that the total variation distance between $\widehat{\Gammab}$ and $\Gammab$ is at most $\epsilon^*/6$, so there is a coupling which further satisfies $\Pr[\widehat{\Gammab}\neq \Gammab\mid \widehat{\Bb}=\Bb]\leq \epsilon^*/6$. Finally, conditioned on $\widehat{\Bb}=\Bb$ and $\widehat{\Gammab}=\Gammab$, we can clearly couple $\hat{\sigmab}$ and $\sigmab$ so that they agree. It follows that the total variation distance between $\hat{\sigmab}$ and $\sigmab$ it at most $2\epsilon^*/3$,  so the result will follow by showing that the distribution of $\sigmab$ and $\mu_{G,H}$ are at distance at most $ 3 e^{-n/(4q)}\leq \epsilon^*/3$.

For this, we will consider the set of configurations $\widehat{\Sigma}:= \Sigma_{G, H,\epsilon}\backslash \Sigma^{\overlap}_{G, H,3\epsilon}$, where recall from Definition~\ref{def:rg5g5gge} that $\Sigma_{G, H,\epsilon}$ is the set of configurations $\tau$ with $\sum_{i\in\{0,1\}}\big|\tau^{-1}(B_i) \cap V_G^i\big| \geq (1 - \epsilon)|V_G|$, and from Definition~\ref{def:zhateps} that $\Sigma^{\overlap}_{G, H,3\epsilon}$  is the set of configurations $\tau$ such that $\tau \in \Sigma^{B_0, B_1}_{G, H,3\epsilon}\cap \Sigma^{C_0, C_1}_{G, H,3\epsilon}$ for some distinct maximal bicliques $(B_0,B_1),(C_0,C_1)$.

Consider arbitrary $\tau\in \widehat{\Sigma}$ and let $(B_0,B_1)$ be such that $\tau\in \Sigma^{B_0,B_1}_{G, H,\epsilon}$. For $i\in \{0,1\}$, let $T_i=V_G^i\cap  \tau^{-1}([q]\backslash B_i)$, and $S_1,\hdots, S_k$ denote the $G^3$-connected components of $T:=T_0\cup T_1$; note, since $\tau\in \widehat{\Sigma}$, we have $|T|\leq \epsilon |V_G|$. Consider the polymer configuration $\Gamma_\tau$ which is the union of the polymers $(S_1,\tau|_{S_1}),\hdots, (S_k,\tau|_{S_k})$. We next find how the sample $\sigmab$ can equal $\tau$; we claim that 
\begin{equation}\label{eq:vtvthbyr5312}
\sigmab =\tau \mbox{ iff } \Bb=(B_0,B_1),\ \Gammab=\Gamma_{\tau},\  \sigmab|_{V_G\backslash (\cup \Gamma)}=\tau|_{V_G\backslash T}.
\end{equation}
The reverse direction of this equivalence is immediate, noting that the equality $\Gammab=\Gamma_{\tau}$ implies that $(\cup \Gamma)=T$ and that $\sigmab|_{\cup \Gammab}=\tau|_{ T}$. So, we focus on the forward direction and assume that $\sigmab=\tau$. Then:
\begin{enumerate}
\item \label{it:a123} $\Bb=(B_0,B_1)$. Otherwise, if $\Bb= (C_0,C_1)$ for some $(C_0,C_1)\neq (B_0,B_1)$, then by Lemma~\ref{lem:smallconfigs} we would have that $\tau\in  \Sigma^{C_0, C_1}_{G, H, 3\epsilon}$, contradicting that $\tau\notin \Sigma^{\overlap}_{G, H,3\epsilon}$.
\item \label{it:b123} $\cup \Gammab=T_0\cup T_1$. Since $\Bb=(B_0,B_1)$ from Item~\ref{it:a123}, for $i\in \{0,1\}$ we have that $u\in V_G^{i} \cap (\cup \Gammab)$ iff $\sigmab(u)\in [q]\backslash B_i$ iff $\tau(u)\in [q]\backslash B_i$ iff $u\in T_i$.  
\item \label{it:c123} $\sigmab|_{\cup \Gammab}=\tau|_{ T}$ and $\sigmab|_{V_G\backslash (\cup \Gamma)}=\tau|_{V_G\backslash T}$. This follows from Item~\ref{it:b123} and $\sigmab=\tau$.
\item $\Gammab=\Gamma_{\tau}$. This follows from Items~\ref{it:b123} and~\ref{it:c123}.
\end{enumerate}

 From \eqref{eq:vtvthbyr5312} and the definition of the sampling procedure for $\sigmab$, we therefore have that $\sigmab=\tau$ with probability
\begin{equation}\label{eq:tt434f3f5}
\frac{|B_0|^n |B_1|^n Z^{B_0, B_1}_{G, H, \epsilon}}{Z_{G, H, \epsilon}^{\text{polymer}}}\cdot \frac{w^{B_0, B_1}_{G, H}(\Gamma_\tau)}{Z^{B_0, B_1}_{G, H, \epsilon}}\cdot \prod_{i\in \{0,1\}}\frac{1}{|B_i|^{|V_G^i\backslash T^+|}} \prod_{u\in \partial T}\Big(\frac{1}{F_u}\prod_{v \in T \cap \partial u} H_{\tau(u), \tau(v)}\Big).
\end{equation}
From \eqref{eq:prodwt} applied to the polymer configuration $\Gamma_\tau$, we have, using that $\cup \Gamma_\tau=T$ and $\sigmab_{\Gamma_\tau}=\tau$, that
\[|B_0|^n |B_1|^n w^{B_0, B_1}_{G, H}(\Gamma_\tau)=\prod_{i\in \{0,1\}}|B_i|^{|V^i_G\backslash T^+|} \prod_{\{u,v\}\in E_G(T)}H_{\tau(u), \tau(v)}\prod_{u \in \partial T}F_u,\]
and hence we obtain that the expression in \eqref{eq:tt434f3f5} equals
\[\frac{\prod_{\{u,v\}\in E_G(T)}H_{\tau(u), \tau(v)}\prod_{u\in \partial T}\prod_{v \in T \cap \partial u} H_{\tau(u), \tau(v)}}{Z_{G, H, \epsilon}^{\text{polymer}}}=\frac{w_{G,H}(\tau)}{Z_{G, H, \epsilon}^{\text{polymer}}},\]
where the last equality follows by noting that edges that are not in $E_G(T)\cup E_G(T,\partial T)$ contribute a factor of 1 in the weight of $T$ (since their endpoints are assigned spins of the biclique). So, we have shown that $\sigmab=\tau$ with probability $w_{G,H}(\tau)/Z_{G, H, \epsilon}^{\text{polymer}}$.

Let $p_{\sigmab}$ be the probability that $\sigmab\in \widehat{\Sigma}$ and $p$ be the aggregate weight of configurations in the Gibbs distribution $\mu_{G,H}$ in $\widehat{\Sigma}$ (that is, $p$ is the probability of seeing a configuration in $\hat{\Sigma}$ when a sample is drawn from~$\mu_{G,H}$). Then, using that $Z^{\overlap}_{G,H,3\epsilon}\leq e^{-n/(3q)}Z_{G,H,\epsilon}$ from Lemma~\ref{lem:zhatepsappxheps} and $Z_{G, H, \epsilon}^{\text{polymer}}\leq (1+e^{-n/(4q)})Z_{G,H,\epsilon}$ from Lemma~\ref{lem:pmerpfappx}, we have that 
\begin{equation}\label{eq:ff555}
p_{\sigmab}\geq \frac{1}{Z_{G, H, \epsilon}^{\text{polymer}}} \sum_{\tau \in \widehat{\Sigma}}w_{G,H}(\tau)\geq \frac{Z_{G,H,\epsilon}-Z^{\overlap}_{G,H,3\epsilon}}{Z_{G, H, \epsilon}^{\text{polymer}}}\geq \frac{\big(1-e^{-n/(3q)}\big)Z_{G,H,\epsilon}}{\big(1+e^{-n/(4q)}\big)Z_{G,H,\epsilon}}\geq 1-2e^{-n/(4q)},
\end{equation}
while for $p$, using that $Z_{G, H}\leq (1+e^{-n})Z_{G,H,\epsilon}$ from Lemma~\ref{lem:gs}, we have the bound
\begin{equation}\label{eq:ff555b}
p\geq \frac{1}{Z_{G, H}} \sum_{\tau \in \widehat{\Sigma}}w_{G,H}(\tau)\geq \frac{Z_{G,H,\epsilon}-Z^{\overlap}_{G,H,3\epsilon}}{Z_{G, H}}\geq \frac{\big(1-e^{-n/(3q)}\big)Z_{G,H,\epsilon}}{(1+e^{-n})Z_{G,H,\epsilon}}\geq 1-2e^{-n/(3q)}.
\end{equation}
It follows that the total variation distance between the distribution of $\sigmab$ and $\mu_{G,H}$ is bounded above by 
\[D:=\frac{1}{2}\big((1-p_{\sigmab})+(1-p)+M\big), \quad \mbox{ where } M:=\sum_{\tau \in \widehat{\Sigma}}w_{G,H}(\tau)\Big| \frac{1}{Z_{G, H, \epsilon}^{\text{polymer}}}-\frac{1}{Z_{G,H}}\Big|.\]
Using Lemma~\ref{lem:gs} and Lemma~\ref{lem:pmerpfappx}, we have the bound 
\begin{equation}\label{eq:vrr545t6}
M\leq  \Big| \frac{Z_{G,H}}{Z_{G, H, \epsilon}^{\text{polymer}}}-1\Big|\leq 2e^{-n/(4q)}.
\end{equation}
Combining \eqref{eq:ff555}, \eqref{eq:ff555b} and \eqref{eq:vrr545t6}, we obtain that $D\leq 3 e^{-n/(4q)}$, i.e., the distance between the distribution of  $\sigmab$ and $\mu_{G,H}$ is at most $3 e^{-n/(4q)}$, as claimed. 

This finishes the proof of Theorem~\ref{thm:ssfpras}.
\end{proof}

\bibliographystyle{plain}
\bibliography{ourbib}
\end{document}